\documentclass[a4paper,11pt]{article}

\usepackage{amsmath,amsthm,amssymb,latexsym,indentfirst,eucal}

\usepackage{newcent}

\usepackage[margin=3cm]{geometry}

\newtheorem{theorem}{Theorem}[section]
\newtheorem{proposition}[theorem]{Proposition}
\newtheorem{lemma}[theorem]{Lemma}

\newtheorem{remark}[theorem]{Remark}

\newtheorem{assumption}[theorem]{Assumption}

\begin{document}

\title{Optimal investment with bounded above utilities in discrete time markets}
\author{Mikl\'os R\'asonyi\thanks{MTA Alfr\'ed R\'enyi Institute of Mathematics, Budapest. 
The author thanks Teemu Pennanen and Ari-Pekka Perkki\"o for 
helpful discussions.}}

\date{\today}

\maketitle

\begin{abstract} 
We consider an arbitrage-free, discrete time and frictionless market. We prove that an investor maximising the
expected utility of her terminal wealth can always find an optimal investment strategy provided 
that her dissatisfaction 
of infinite losses is infinite and her utility function is non-decreasing, continuous and bounded above. 
The same result is shown for
cumulative prospect theory preferences, under additional assumptions.
\end{abstract}

\textbf{MSC (2010):} Primary 93E20, 91G10; Secondary
91B06.

\smallskip

\textbf{Keywords:} expected utility; utility maximisation; cumulative prospect theory; 
bounded above utility functions

\section{Introduction}

When deciding between two investments with random returns $X$, $Y$, the prevailing approach 
in economic theory is to compare their 
expected utilities
$Eu(X)$, $Eu(Y)$ where $u$ is a utility function describing the preferences of the decision-maker. 

In most studies, the function $u$ 
is non-decreasing and concave. While 
the former property (preferring more money to less) could hardly be contested, the latter one (corresponding to
an assumed risk-averse behaviour of the given agent) has been challenged on an empirical basis, see
e.g. \cite{kt,tk}. It has furthermore been argued there that decisions are often based on a distorted view of
events' probabilities.

In the present paper we prove a surprisingly general existence theorem for optimal strategies in a multiperiod 
market model where  
investors' preferences are not necessarily concave and may be based on distorted beliefs. 
We can accomodate piecewise concave or S-shaped (convex up to a certain point, then concave)
utility functions. The crucial assumption we make
is that $u$ should be bounded from above. In the setting of a concave $u$ very general results have
been obtained in \cite{pennanen1} which apply to markets with frictions as well. 

In the economics literature there has been a considerable amount of controversy over bounded utilities. 
Starting with \cite{menger34} and up to recently (see e.g. \cite{muravievrogers13}), it has been 
asserted that $u$ should be bounded from above as well as from below.
Counterarguments also emerged over the time, see e.g. \cite{ryan}.
We will see that a bound from 
below often excludes a meaningful optimisation problem (see Remark \ref{meaning} below), 
hence we avoid such an assumption. 

We think that in the present, discrete-time setting it is reasonable
to allow $u$ unbounded above as well as below. Nonetheless, in the present article 
we restrict our attention  
to utilities which are
bounded above. This simplifies arguments 
and allows to obtain clear-cut results (Theorems \ref{elso} and \ref{masodik}).

Most of related previous work on non-concave utilities concentrated on the case of complete market models,
\cite{bkp,jz,cp,cd,reichlinthesis,rr,r12,rr2}, only very specific continuous-time incomplete
models could be treated, \cite{reichlinthesis,rr}. One-step models were studied in
\cite{bg,hz}. For multistep, incomplete models \cite{cr,nc,rrv} are the only 
references we are aware of.
These latter papers proved existence theorems for $u$ possibly unbounded
but in the current setting of bounded above utilities the present article
obtains much sharper results, see Remarks
\ref{beyond} and \ref{beyond1}.

We consider a probability space $(\Omega,\mathcal{F},P)$ with a discrete-time filtration
$\mathcal{F}_t$, $t=0,\ldots,T$ with $T\geq 1$ and $\mathcal{F}_0$ trivial. All sigma-algebras in this
article are assumed $P$-complete. The prices of $d$ risky assets are described by the
adapted $\mathbb{R}^d$-valued process $S_t$, $t=0,\ldots,T$, set $\Delta S_j:=S_{j}-S_{j-1}$, $1\leq j\leq T$. 
Investments in the
respective assets are represented by $\mathbb{R}^d$-valued processes $\phi_t$,
$t=1,\ldots,T$, which are predictable (i.e. $\phi_t$ is $\mathcal{F}_{t-1}$-measurable).
The class of all trading strategies is denoted by $\Phi$. The riskless asset is assumed
to have unit price at all times, hence the wealth $X_t^{\phi,z}$ of a portfolio starting from initial
capital $z$ and pursuing strategy $\phi$ equals
$$
X^{\phi,z}_t=z+\sum_{j=1}^t \phi_j\Delta S_j
$$
at time $t=0,\ldots,T$.

For a random variable $B$ and initial wealth $z$, we will first look at the problem of finding $\phi^*\in\Phi$ saisfying
\begin{equation}\label{maci}
Eu(X^{\phi^*,z}_T-B)=\sup_{\phi\in\Phi}Eu(X^{\phi,z}_T-B)
\end{equation}
for some $u:\mathbb{R}\to\mathbb{R}$ measurable and bounded above. Note that in this
case the expectations (as well as related conditional expectations) are well-defined but may
take the value $-\infty$.

The quantity $\bar{u}(z):=\sup_{\phi\in\Phi}Eu(X^{\phi,z}-B)$,
$z\in\mathbb{R}$ is the indirect utility of initial capital $z$: this is the maximal
satisfaction that the given investor can derive from $z$ using the trading opportunities available
in the market.

Depending on the context, the random variable $B$ may have various interpretations: 
it can represent the payoff of a contingent
claim which the investor has to deliver at $T$ or it can be a reference point (benchmark), such as
the terminal wealth of another investor acting in a (perhaps different) market. 

We say that there is no arbitrage (NA) if, for all $\phi\in\Phi$,
$$
X^{\phi,0}_T\geq 0\mbox{ a.s.}\Rightarrow X^{\phi,0}_T=0\mbox{ a.s.}
$$
If (NA) fails and $u$ is strictly increasing then no optimal strategy $\phi^*$ exists for
the problem \eqref{maci} since $\phi^*+\phi$ would always outperform $\phi^*$ for any
$\phi$ violating (NA). Hence it is natural to assume (NA).

In Section \ref{multi} we will prove the following result:

\begin{theorem}\label{elso} Assume (NA). Let $u:\mathbb{R}\to\mathbb{R}$ be a nondecreasing and continuous function  
which is bounded above and
satisfies 
\begin{equation}\label{limma}
\lim_{x\to -\infty}u(x)=-\infty.           
\end{equation}
 Let $B$ be an arbitrary real-valued random variable
with 
\begin{equation}\label{ccc}
Eu(z-B)>-\infty\mbox{ for all }z\in\mathbb{R}.
\end{equation}
Then for all $z\in\mathbb{R}$ there exists a strategy $\phi^*=\phi^*(z)\in\Phi$ such that
$$
Eu(X_T^{z,\phi^*}-B)=\sup_{\phi\in\Phi}Eu(X_T^{z,\phi}-B)=\bar{u}(z).
$$
Furthermore, $\bar{u}$ is continuous on $\mathbb{R}$.
\end{theorem}

Continuity of $u$ is a natural requirement. Theorem
\ref{elso} guarantees that the indirect utility $\bar{u}$ inherits the continuity property of $u$.
Condition \eqref{limma} means that infinite losses lead to an infinite dissatisfaction of the agent.

\begin{remark}\label{beyond} {\rm In Corollary 2.12 of \cite{nc} the existence result 
Theorem \ref{elso} has been obtained 
(in the case $B=0$) 
under the following additional
assumptions: 
\begin{enumerate}

 \item $E(u(x+y\Delta S_t)|\mathcal{F}_{t-1})>-\infty$ for all $x\in\mathbb{R}$, $y\in\mathbb{R}^d$ and $t=1,\ldots,T$;
 $u(0)=0$;

 \item there exist $\underline{x}<0$ and $\underline{\gamma}>0$ such that $u(\underline{x})<0$ and
 for all $x\leq \underline{x}$ and $\lambda\geq 1$, $u(\lambda x)\leq \lambda^{\underline{\gamma}} u(x)$.

\end{enumerate}
While 1. looks rather harmless, 2. forces $u$ to decrease to $-\infty$ at least at the speed
of a power function and hence it excludes the cases where $u(x)$ behaves like e.g. $-\ln(-x)$ near
$-\infty$. This shows that Theorem \ref{elso} indeed extends available results in a significant way.}
\end{remark}

\begin{remark}\label{meaning}
 {\rm Let us take $u$ bounded above, continuously differentiable and define $S_0=0$, $S_1=\pm 1$ with probabilities
 $1/2-1/2$. If $u'(\phi)-u'(-\phi)>0$ for all $\phi>0$ (this is easily seen to imply
 $u(-\infty)>-\infty$) then $\phi\to Eu(\phi\Delta S_1)$ is
 strictly increasing in $|\phi|$ which excludes the existence of an optimiser $\phi^*$ 
 for \eqref{maci}. 
 
 This shows that -- even for very simple specifications of $S_0,S_1$ -- the failure of \eqref{limma}
 may easily prevent the existence of an optimiser, highlighting the importance of condition \eqref{limma}.}
 \end{remark}

Our second main result proves the existence of optimal investments for investors following the principles
of cumulative prospect theory.
We assume that
$u_{\pm}:\mathbb{R}_+\to\mathbb{R}_+$ and $w_{\pm}: [0,1]\to [0,1]$  are
continuous non-decreasing functions such that
$u_{\pm}(0)=0$, $w_{\pm}(0)=0$ and $w_{\pm}(1)=1$.
We fix $B$, a scalar-valued random variable. The agent's utility
function will be $u(x):=u_+(x)$, $x\geq 0$, $u(x):=-u_-(-x)$, $x<0$. We assume $u_+$ bounded above
(resp. $u_-(\infty)=\infty$) which guarantee that $u$ is bounded above
(resp. $u(-\infty)=\infty$). The functions $w_{+}$ (resp. $w_-$) will represent probability
distortions applied to gains (resp. losses) of the investor. For $x\in\mathbb{R}$ we denote
by $x_+$ (resp. by $x_-$) the positive (resp. the negative) part of $x$.

We define, for $\theta\in\Phi$,
\begin{eqnarray*}
V^+(\theta,z):=\int_0^{\infty} w_+\left(P\left(
u_+\left(\left[X^{\theta,z}_T-B\right]_+\right)\geq y
\right)\right)dy,
\end{eqnarray*}
and
\begin{eqnarray}\label{cc}
V^-(\theta,z):=\int_0^{\infty} w_-\left(P\left(
u_-\left(\left[X^{\theta,z}_T-B\right]_-\right)\geq y
\right)\right)dy,
\end{eqnarray}
and set $V(\theta,z):=V^+(\theta,z)-V^-(\theta,z)$.

We aim to find $\theta^*\in\Phi$ with
\begin{equation}\label{problem}
 V(\theta^*,z)=\sup_{\theta\in\Phi}V(\theta,z).
\end{equation}

Note that if $w_{\pm}(p)=p$ (that is, there is no distortion) then we have 
$V(\theta,z)=Eu(X_T^{\phi,z})$, hence problem \eqref{maci} is a subcase of problem \eqref{problem}.

We need the following technical conditions.

\begin{assumption}\label{AAA}
Let $\mathcal{F}_t$ be the $P$-completion of $\sigma(Z_1,\ldots,Z_t)$ for $t=1,
\ldots,T$,
where the $Z_i$, $i=1,\ldots, T$ are $\mathbb{R}^N$-valued
independent random variables. $S_0$ is constant
and $S_1=f_1(Z_1)$, $S_t=f_t(S_1,\ldots,S_{t-1},Z_t)$, $t=2,\ldots,T$ for some continuous functions $f_t$ 
(hence $S_t$ is adapted). We assume that $B=g(S_1,\ldots,S_T)$ for some continuous $g$.

Furthermore, for $t=1,\ldots,T$ there exists an $\mathcal{F}_t$-measurable
uniformly distributed random variable $U_t$ which is independent of $\mathcal{F}_{t-1}\vee \sigma (S_t)$.
\end{assumption}

\begin{remark}
{\rm Stipulating the existence of the ``innovations'' $Z_t$ might look restrictive but it can be weakened 
to $(Z_1,\ldots,Z_T)$ having a nice enough density w.r.t. the respective Lebesgue measure,
see Proposition 6.4 of \cite{cr}. In addition to the continuity conditions, the above Assumption
requires that the information filtration is ``large enough'': at each time $t$, there should exist
some randomness which is independent of both the past ($\mathcal{F}_{t-1}$) and of the present price ($S_t$).
Since real markets are perceived as highly incomplete and noisy, this looks a mild requirement.
See Section 8 of \cite{cr} for models satisfying Assumption \ref{AAA}.}
\end{remark}

The following result will be shown in Section \ref{uma}.
\begin{theorem}\label{masodik} Assume that 
\begin{equation}\label{cccc}
V^-(0,z)>-\infty\mbox{ for all }z\in\mathbb{R}. 
\end{equation}
Under Assumption \ref{AAA} and (NA), for each $z\in\mathbb{R}$  there is $\theta^*=\theta^*(z)\in\Phi$ satisfying
\eqref{problem}.
\end{theorem}

\begin{remark}\label{beyond1} {\rm In Theorem 7.4 of \cite{cr} the existence result Theorem \ref{masodik} 
has been obtained 
under the following additional
assumptions: 
\begin{enumerate}

 \item $w_+$ (resp. $w_-$) is dominated by (resp. dominates)  a power function;
 
 \item $u_-$ dominates a power function;
 
 \item integrability conditions on $S_t,1/\beta_t,1/\kappa_t$, see Proposition \ref{karakter} below
 for the definition of $\beta_t,\kappa_t$.
 
 \end{enumerate}
Again, in the case of $u_+$ bounded above, our assumptions are much weaker.}
\end{remark}

\begin{remark}
{\rm In continuous-time models both Theorem \ref{elso} and Theorem \ref{masodik} fail and one needs 
further assumptions
on $u_-,w_-$ in order to get existence, see \cite{rr2}.} 
\end{remark}


\section{Expected utility maximisation -- one-step case}

In this section we consider a function $V:\Omega\times \mathbb{R}\to\mathbb{R}$ such that, for every
$\omega\in\Omega$, the functions $x\to V(\omega,x)$ are non-decreasing, continuous and they
satisfy $\lim_{x\to-\infty}V(\omega,x)=-\infty$ and $V(\omega,x)\leq C$, for all $x\in\mathbb{R}$, 
with some fixed constant $C$. 

We assume that there is a 
sequence of real-valued random variables $M(n)$, $n\in\mathbb{Z}$ such that $V(n)\geq M(n)$ holds a.s.,
for all $n$. Let $\mathcal{H}\subset\mathcal{F}$ be a sigma-algebra. We assume that,
for all $n$, $m(n):=E(M(n)|\mathcal{H})>-\infty$ a.s.

Let $Y$ be an $\mathbb{R}^d$-valued random variable and let
$D(\omega)$ be the affine hull of the support of a regular version of $P(Y\in\cdot|\mathcal{H})(\omega)$.
We can choose $D\in\mathcal{H}\otimes\mathcal{B}(\mathbb{R}^d)$, see Proposition A.1 of \cite{rs}.

The notation $\Xi^n$ will be used for the class of $\mathbb{R}^n$-valued $\mathcal{H}$-measurable
random variables. For $\xi\in\Xi^d$ we will denote by $\hat{\xi}(\omega)$ the projection of $\xi(\omega)$
on $D(\omega)$. The function $\omega\to\hat{\xi}(\omega)$ is $\mathcal{H}$-measurable, see Proposition 4.6
of \cite{rs}. 

\begin{remark}\label{projection} {\rm Notice that $P\left((\xi-\hat{\xi})Y=0|\mathcal{H}\right)=1$ a.s., hence
$P(\xi Y=\hat{\xi} Y)=1$. It means that we may always replace $\xi$ by $\hat{\xi}$ in
the maximisation of $E(V(x+\xi Y)|\mathcal{H})$ in $\xi$, which is the subject of the
present section.}
\end{remark}

We assume that there exist $\mathcal{H}$-measurable $\kappa,\beta>0$ such that for all
$\xi\in\Xi^d$ with $\xi(\omega)\in D(\omega)$ a.s. we have
\begin{equation}\label{nna}
P(\xi Y\leq -\beta |\xi|\,\vert\mathcal{H})\geq \kappa\mbox{ a.s.}
\end{equation}

Following Lemma 4.8 of \cite{rs} and Lemma 3.11 of \cite{nc}, the  next lemma allows
to work with strategies admitting a fixed bound.

\begin{lemma}\label{keyy}
There exists $v:\Omega\times \mathbb{R}\to\mathbb{R}$ such that, for all $x$, $v(\omega,x)$ is a version of
$\mathrm{ess.}\sup_{\xi\in\Xi^d}E(V(x+\xi Y)\vert\mathcal{H})$ and, for every
$\omega\in\Omega$, the functions $x\to v(\omega,x)$ are non-decreasing, right-continuous,
they satisfy $\lim_{x\to-\infty}v(\omega,x)=-\infty$ and $v(\omega,x)\leq C$, for all $x\in\mathbb{R}$.
There exist random variables $K(n)$, $n\in\mathbb{Z}$ such that for all $x\in\mathbb{R}$ and 
$\xi\in\Xi^d$ with $\xi\in D$ a.s.,
\begin{equation}\label{morgiana}
E(V(x+\xi Y)\vert\mathcal{H})\leq E(V(x+ 1_{|\xi|\leq K([x])}\xi  Y)\vert\mathcal{H}),
\end{equation}
where $[x]$ denotes the integer part of $x$. We have $m(n)\leq v(n)$ a.s., for all $n$.
\end{lemma}
\begin{proof} 
Fix $n\in\mathbb{Z}$.
Since $V(y)\to -\infty$ a.s. when $y\to -\infty$, for each $L\in\Xi^1$ there is $G_L\in\Xi^1$ such that
$P(V(-G_L)\leq -L|\mathcal{H})\geq 1-\kappa/2$ a.s. Since for all $x\in [n,n+1)$,
\begin{eqnarray*}
E(V(x+\xi Y)|\mathcal{H})\leq C+ E(V(n+1-|\xi|\beta)1_{\{\xi Y\leq - |\xi|\beta,\,V(-G_L)\leq -L\}}|\mathcal{H}) 
\end{eqnarray*}
and, by \eqref{nna}, $P(\xi Y\leq -\beta |\xi|\,|\mathcal{H})\geq\kappa$ a.s., 
we get that whenever $\beta |\xi|\geq G_L+n+1$, we have
\begin{eqnarray*}
E(V(x+\xi Y)\vert\mathcal{H})\leq C+ E(V(n+1-|\xi|\beta)1_{\{\xi Y\leq -|\xi|\beta,\,V(-G_L)\leq -L\}}\vert\mathcal{H})\leq
C-L(\kappa/2).
\end{eqnarray*}
Choose $L:=2(C-m(n))/\kappa$, then $K(n):=(G_L+n+1)/\beta$ is such that for $|\xi|\geq K(n)$,
$$
E(V(x+\xi Y)|\mathcal{H})\leq m(n)\leq E(V(x)|\mathcal{H})
$$
holds a.s., providing a suitable function $K(\cdot)$.

Now for each $x$, let $F(x)$ be an arbitrary version of the essential supremum in consideration.
 We may and will assume that $m(n)\leq F(x)\leq C$ for all $\omega$ and $x\in [n,n+1)$, 
 for all $n$.
 Outside a negligible set $N\subset\Omega$, $q\to F(q)$ in non-decreasing on $\mathbb{Q}$.
 Define $v(x):=\inf_{q>x,q\in\mathbb{Q}}F(q)$. This function is non-decreasing and right-continuous outside $N$.  
 
Fix $x\in\mathbb{R}$. Since $F(x)\leq F(q)$ a.s. for $x\leq q$, we clearly have $F(x)\leq v(x)$ a.s. Let us now take 
 $\xi_k\in\Xi^d$ with
 $$
 F(q_k)-1/k\leq E(V(q_k+\xi_k Y)\vert\mathcal{H}), 
 $$
a.s. where $x<q_k<x+1$, $q_k\in\mathbb{Q}$ decreases to $x$ as $k\to\infty$.
Then, by Remark \ref{projection} and \eqref{morgiana},
$$
F(q_k)-1/k\leq E(V(q_k+\bar{\xi}_k Y)|\mathcal{H}),
$$
where $\bar{\xi}_k:=\hat{\xi}_k 1_{|\hat{\xi}_k|\leq K([x])+K([x+1])}$, $k\in\mathbb{N}$.
By Lemma 2 of \cite{ks}, there is an $\mathcal{H}$-measurable random 
subsequence\footnote{That is, $\mathcal{H}$-measurable random variables 
$k_n:\Omega\to\mathbb{N}$, $n\in\mathbb{N}$ with $k_n(\omega)<k_{n+1}(\omega)$ for all $\omega\in\Omega$
and $n\in\mathbb{N}$.} $k_n$, $n\to\infty$ such 
that $\bar{\xi}_{k_n}\to\xi'$ a.s.
with some $\xi'\in\Xi^d$. The Fatou lemma and continuity of $V$ imply that 
$$
v(x)=\lim_{n\to\infty} (F(q_{n_k})-1/n_k)\leq E(V(x+\xi' Y)\vert\mathcal{H})\leq  F(x) 
$$
a.s., showing that $v(x)$ is indeed a version of $F(x)$.

Finally we claim that $v(x)\to-\infty$, $x\to-\infty$ a.s. For each $L(n):=n$
there is $G_{L(n)}\in\Xi^1$ with $P(V(-G_{L(n)})\leq -L(n)|\mathcal{H})\geq 1-\kappa/2$. 
Let us notice that 
\begin{eqnarray*}
E(V(-G_{L(n)}+\xi Y)|\mathcal{H})\leq E(V(-G_{L(n)})1_{\{\xi Y\leq 0,\,V(-G_{L(n)})\leq -L(n)\}}\vert\mathcal{H})+C \leq -n \kappa/2 +C,
\end{eqnarray*}
which is a bound independent of $\xi$ and it tends to $-\infty$ a.s. as $n\to -\infty$,
so $v(-G_{L(n)})\to -\infty$.
By the monotonicity of $v$ this implies that a.s. $\lim_{x\to -\infty}v(x)=-\infty$.
Hence our claim follows. By modifying $v$ on a null set we get
that the claimed properties hold for each $\omega\in\Omega$. The last statement of this Lemma is trivial.
\end{proof}

\begin{lemma}\label{kullogo}
Let $H\in \Xi^1$. Then $v(H)$ is
a version of $\mathrm{ess.}\sup_{\xi\in\Xi^d}E(V(H+\xi Y)\vert\mathcal{H})$.
\end{lemma}
\begin{proof} 
Working separately on the events $\{ H\in [n,n+1)\}$ we may and will assume that
$H\in [n,n+1)$ for a fixed $n$. The statement is clearly true for constant $H$ by
Lemma \ref{keyy} and hence also for countable step functions $H$. For general $H$,
let us take step functions $H_k\in [n,n+1)$, $H_k\in\Xi^1$ decreasing to $H$ as $k\to\infty$.
$v(H_k)\to v(H)$ a.s. by right-continuity. It is also clear that, for all $\xi\in\Xi^d$,
$E(V(H+\xi Y)\vert\mathcal{H})\leq E(V(H_k+\xi Y)\vert\mathcal{H})$ a.s. for all $k$,
hence 
$$
\mathrm{ess.}\sup_{\xi\in\Xi^d}E(V(H+\xi Y)\vert\mathcal{H})\leq v(H).
$$
Choose $\xi_k$ with $E(V(H_k+\xi_k Y)\vert\mathcal{H})>v(H_k)-1/k$ and 
note that $\bar{\xi}_k:=\hat{\xi}_k 1_{|\hat{\xi}_k|\leq K(n)}$ also satisfies
$E(V(H_k+\bar{\xi}_k Y)\vert\mathcal{H})>v(H_k)-1/k$ by Remark \ref{projection} and Lemma \ref{keyy}.

By Lemma 2 of \cite{ks} we can take
an $\mathcal{H}$-measurable random subsequence $k_l$, $l\to\infty$ such that
$\bar{\xi}_{k_l}\to \xi^{\dagger}$ a.s., $l\to\infty$. Fatou's lemma and continuity of $V$ imply
$$
v(H)=\lim_{l\to\infty}[v(H_{k_l})-1/k_l]\leq\limsup_{l\to\infty}E(V(H_{k_l}+\bar{\xi}_{k_l} Y)\vert\mathcal{H})\leq
E(V(H+\xi^{\dagger} Y)\vert\mathcal{H}),
$$
completing the proof.
\end{proof} 

\begin{lemma}\label{qef}
Outside a negligible set, the trajectories $x\to v(x,\omega)$ are continuous.
\end{lemma}
\begin{proof}
Arguing by contradiction, let us suppose that the projection $A\in\mathcal{H}$ of the set
$$
B:=\{(x,\omega):v(x,\omega)>\varepsilon +\sup_{q<x,q\in\mathbb{Q}} v(q,\omega)\}\in\mathcal{B}(\mathbb{R})\otimes\mathcal{H}
$$
on $\Omega$ has positive probability for some $\varepsilon>0$. 
Let $H$ be a measurable selector of $B$ on $A$ and let it be $0$ outside $A$. 
Let $H>H_k$ be rational-valued $\mathcal{H}$-measurable step functions
increasing to $H$ and choose $\zeta_l$ such that $v(H)-1/l\leq E(V(H+\zeta_l Y)|\mathcal{H})$
for each $l\in\mathbb{N}$. On $A$ we have 
$$
\limsup_{k\to\infty}E(V(H_k+\zeta_l Y)|\mathcal{H})\leq
\limsup_{k\to\infty}v(H_k)\leq v(H)-\varepsilon.
$$
On the other hand, monotone convergence ensures $\lim_{k\to\infty} E(V(H_k+\zeta_l Y)|\mathcal{H})=
E(V(H+\zeta_l Y)|\mathcal{H})\geq v(H)-1/l$, for all $l$. Tending $l\to\infty$ leads to a contradiction.
\end{proof}


By modifying $v(\cdot,\omega)$ on a negligable set we may and will assume that $v(\cdot,\omega)$
is continuous for all $\omega\in\Omega$.

\begin{lemma}\label{farosz} For each real-valued $\mathcal{H}$-measurable random variable $H$
there exists $\mathcal{H}$-measurable $\xi^*_H:\Omega\to\mathbb{R}^d$ such that
$E(V(H+\xi^*_H Y)\vert\mathcal{H})=v(H)$ a.s.
 \end{lemma}
\begin{proof} The set $E(V(H+\xi Y)|\mathcal{H})$, $\xi\in\Xi^d$ is directed upwards hence
there is a sequence $\xi_k\in\Xi^d$ such that 
$E(V(H+\xi_k Y)|\mathcal{H})$ increases to $\mathrm{ess.}\sup_{\xi\in\Xi^d}E(V(H+\xi Y)|\mathcal{H})=v(H)$,
recall Lemma \ref{kullogo}. 
Define the random variable $K([H]):=\sum_{n\in\mathbb{Z}}K(n)1_{H\in [n,n+1)}$.
By Remark \ref{projection} and Lemma \ref{keyy}, for $\bar{\xi}_k:=\hat{\xi}_k 1_{|\hat{\xi}_k|\leq K([H])}$,
$$
E(V(H+\xi_k Y)|\mathcal{H})\leq E(V(H+\bar{\xi}_k  Y)\vert\mathcal{H})
$$
a.s. for all $k$. Then an $\mathcal{H}$-measurable 
random subsequence $k_n$ exists such that $\bar{\xi}_{k_n}\to\xi^*$ a.s., $n\to\infty$
for some $\xi^*\in\Xi^d$, and Fatou's lemma
guarantees that
$$
E(V(H+\xi^* Y)|\mathcal{H})\geq \limsup_{n\to\infty}E(V(x+\bar{\xi}_{k_n}  Y)\vert\mathcal{H})\geq
\lim_{n\to\infty} E(V(x+\xi_n Y)|\mathcal{H})=v(H),
$$
hence $\xi^*_H:=\xi^*$ is as required.
\end{proof}
  
\section{The multi-step case}\label{multi}

We first recall a useful consequence of (NA). Let $D_t(\omega)$ be the affine hull of the support of a regular version of
$P(\Delta S_{t+1}\in\cdot |\mathcal{F}_t)(\omega)$. Let $\Xi_t^n$ denote the class of $\mathcal{F}_t$-measurable
$n$ dimensional random variables.
\begin{proposition}\label{karakter}
(NA) implies the existence of 
$\beta_t,\kappa_t\in\Xi^1_t$ with $\beta_t,\kappa_t>0$ a.s.
such that for  $\xi \in {\Xi}_t^d$ with  $\xi \in D_{t}$ a.s.:
\begin{equation}\label{valaki}
P( \xi \Delta S_{t+1} \leq -\beta_t |\xi| \vert\mathcal{F}_t)\geq \kappa_t\quad\mbox{ on }\{D_{t}\neq \{0\}\}
\end{equation}
holds almost surely; for all $0\leq t\leq T-1$.
\end{proposition}
\begin{proof} Trivial from Proposition 3.3 of \cite{rs}.
\end{proof}

\begin{remark}\label{pop}
{\rm Note that if $Q\sim P$ then \eqref{valaki} implies that
$$
Q( \xi \Delta S_{t+1} \leq -\beta_t |\xi| \vert\mathcal{F}_t)\geq \kappa^Q_t\quad\mbox{ on }\{D_{t}\neq \{0\}\}
$$
for some $\kappa^Q_t>0$ a.s.}
\end{remark}

In this section all the assumptions of Theorem \ref{elso} will be in force.
Set $U_T(x,\omega):=u(x-B(\omega))$ for $(\omega,x)\in \Omega\times\mathbb{R}^d$.

\begin{lemma}\label{potty}
For $t=0,\ldots,T-1$, there exist 
$U_t:\Omega\times\mathbb{R}\to\mathbb{R}$ such that, for all $x$, $U_t(\omega,x)$ is a version of
$\mathrm{ess.}\sup_{\xi\in\Xi^d}E(U_{t+1}(x+\xi \Delta S_{t+1})\vert\mathcal{F}_t)$ and, for every
$\omega\in\Omega$, the functions $x\to U_t(\omega,x)$ are non-decreasing, continuous,
they satisfy $\lim_{x\to-\infty}U_t(\omega,x)=-\infty$ and $U_t(\omega,x)\leq C$, for all $x\in\mathbb{R}$,
with some constant $C$.
For each $\mathcal{F}_t$-measurable $H_t$ there exists
$\tilde{\xi}_t(H_t)\in\Xi^d_t$ such that
$$
E(U_{t+1}(H_t+\tilde{\xi}_t(H_t) \Delta S_{t+1})\vert\mathcal{F}_t)=
\mathrm{ess.}\sup_{\xi\in\Xi^d}E(U_{t+1}(H_t+\xi \Delta S_{t+1})\vert\mathcal{F}_t).
$$
\end{lemma}
\begin{proof}
 Proceeding by backward induction, we will show the statements of this Lemma together with the
 existence of $M_t(n)\leq U_t(n)$ with $EM_t(n)>-\infty$. 
 First apply Lemmata \ref{keyy}, \ref{qef} and Proposition \ref{farosz}  to $Y:=\Delta S_T$, $V:=U_T$,
 $D:=D_{T-1}$ and $M(n)=M_T(n):=U_T(n)$  and we get the statements for $T-1$. 
 Note that (NA) implies \eqref{nna} for $Y$ by Proposition \ref{karakter}. Also, $M_{T-1}(n):=E(U_T(n)|\mathcal{F}_{T-1})$
 satisfies $EM_{T-1}(n)>-\infty$ by \eqref{ccc}. 
 
 Assume that this Lemma has been shown for $t+1$.
 Apply Lemmata \ref{keyy}, \ref{qef} and \ref{farosz} with the choice $Y:=\Delta S_{t+1}$,
 $V:=U_{t+1}$, $D:=D_{t}$ and $M(n):=M_{t+1}(n)$ to get the statements for $t$, noting that for 
 $m(n)=M_t(n):=E(M_{t+1}(n)|\mathcal{F}_t)=E(U_T(n)|\mathcal{F}_{t})$, $Em(n)>-\infty$.
\end{proof}

\noindent{\em Proof of Theorem \ref{elso}.}
 Using the previous Lemma, define recursively $\phi^*_1:=\tilde{\xi}_1(z)$ and
 $\phi^*_{t+1}:=\tilde{\xi}_{t+1}(X^{\phi^*,z}_t)$. For any $\phi\in\Phi$:
 \begin{eqnarray*}
 Eu(X_T^{\phi,z}-G) &=& EE(U_T(X_{T-1}^{\phi,z}+\phi_T\Delta S_T)|\mathcal{F}_{T-1})\leq 
 EU_{T-1}(X_{T-1}^{\phi,z})\\ &\leq& \ldots\leq EU_0(z)
 \end{eqnarray*}
by repeated applications of Lemma \ref{kullogo} for $v=U_t$, $t=T-1,\ldots,0$.
Notice that there are equalities everywhere for $\phi^*$. This finishes the proof
except the continuity of $\bar{u}$. Note that $U_0$ is continuous outside a negligible set
by Lemma \ref{potty}. Clearly, $\bar{u}(x)=EU_0(x)$. Let $x_n$ converge to $x$ and let $k\in\mathbb{Z}$ 
be such that $k\leq\inf_n x_n$. Then $U_0(x_n)\to U(x)$
a.s. and $M_0(k)\leq U_0(x_n)\leq C$ where $C$ is an upper bound for $u$. Since $M_0(k)$ is integrable,
dominated convergence finishes the proof.\hfill $\Box$

\section{Utility maximisation under cumulative prospect theory}\label{uma}

We begin with a variant of Lemma A.1 in \cite{cr}. Let us define 
$$
\hat{\Phi}:=\{\theta\in\Phi:
\theta_t\in D_t\mbox{ a.s.},\ t=1,\ldots,T\}.
$$

\begin{lemma}\label{makk} For all $t=0,\ldots,T-1$,
there exist $\pi_t\in\Xi^1_t$, ${\pi}_t>0$ a.s., $t=0,\ldots,T-1$,
such that, for all $\theta\in\hat{\Phi}$, 
\[
P(\theta_{t+1} \Delta S_{t+1}\leq -\kappa_t \vert \theta_{t+1}\vert,\
\theta_n\Delta S_n\leq 0,\, n=t+2, \ldots, T| \mathcal{F}_t)\geq{\pi}_t
\quad\mbox{ on }\{D_{t}\neq \{0\}\}.
\]
\end{lemma}
\begin{proof}
Define the events
\begin{eqnarray*}
A_{t+1} & :=& \{ \theta_{t+1}\Delta S_{t+1}\leq -\kappa_t \vert \theta_{t+1}\vert\},\\
A_n & := & \{{\theta}_n\Delta S_n\leq 0\},\quad t+2\leq n\leq T.
\end{eqnarray*}
We will prove, by induction on $m=t+1,\ldots,T$, that for all $Q\sim P$, there is $\pi^Q_t(m)>0$ a.s.
such that
\begin{eqnarray}\label{hypo}
E_Q(1_{A_{t+1}}\ldots 1_{A_m}| \mathcal{F}_t)\geq \pi_t^Q(m)\quad\mbox{ on }\{D_{t}\neq \{0\}\}.
\end{eqnarray}
For $m=t+1$ this is trivial for $P=Q$ by Proposition \ref{karakter} and it follows
for all $Q\sim P$ by Remark \ref{pop}.

Let us assume that \eqref{hypo} has been shown for $m-1$, we will establish it for $m$.
\begin{eqnarray*}
E_Q(1_{A_m}\ldots 1_{A_{t+1}}| \mathcal{F}_t) & = & E_Q(E_Q(1_{A_m}| \mathcal{F}_{m-1})
1_{A_{m-1}}\ldots 1_{A_{t+1}}| \mathcal{F}_t)\\
& \geq & E_Q(\kappa_{m-1}^Q 1_{A_{m-1}}\ldots 1_{A_{t+1}}| \mathcal{F}_t).
\end{eqnarray*}
Now define $R\sim P$ by $dR/dP:=\kappa^Q_{m-1}/E\kappa^Q_{m-1}$. It follows that
\begin{eqnarray*}
& & E_Q(\kappa_{m-1}^Q 1_{A_{m-1}}\ldots 1_{A_{t+1}}| \mathcal{F}_t) =
E_Q(\kappa_{m-1}^Q) E_R(1_{A_{m-1}}\ldots 1_{A_{t+1}}| \mathcal{F}_t) E_Q(dR/dQ|\mathcal{F}_t)\\ &\geq&
E_Q(\kappa_{m-1}^Q)E_Q(dR/dQ|\mathcal{F}_t)\pi^R_{t}(m-1)>0\mbox{ a.s.}\quad\mbox{ on }\{D_{t}\neq \{0\}\},
\end{eqnarray*}
showing the induction step. Finally, one can set $\pi_t:=\pi^P_t(T)$.
\end{proof}

\begin{lemma}\label{maine} Let $\theta(n)\in\hat{\Phi}$, $n\in\mathbb{N}$ such that 
$V(z,\theta(n))\geq -c$ for some $c\in\mathbb{R}$, for all $n$. Then
the sequence of the laws of $(\theta_1(n),\ldots,\theta_T(n))$ is tight.
\end{lemma}
\begin{proof} It suffices to show, by induction on $m=1,\ldots,T$, that we have, for all $m$,
$$
P(|\theta_m(n)|\geq c_n)\to 0,\ n\to\infty,
$$
for every $0\leq c_n\to\infty$, $n\to\infty$, see e.g. Lemma 4.9 on p. 66 of \cite{kallenberg}.
The first step is similar to the induction step, so we omit it. Let us assume that
the above statement has been shown for $m=1,\ldots,k$, we will show it for $k+1$.

Define the right-continuous and non-decreasing function 
$w_-^{-1}(q):=\max\{p\in [0,1]:w_-(p)=q\}$, $q\in [0,1]$ 
(note the continuity of $w_-$). 
Clearly, $w_-^{-1}(q)\to 0$, $q\to 0$.

Since $u_+$ and hence also $V^+(\theta(n),z)$ are bounded above by a fixed constant $C$,
$V(\theta(n),z)\geq -c$ implies $V^-(\theta(n),z)\leq c+C$ for all $n$.
From \eqref{cc} we get $w_-(P((u_-(X_T^{\theta(n),z}-B)_-)\geq y))\leq (c+C)/y$ hence also 
$P(u_-((X_T^{\theta(n),z}-B)_-)\geq y)\leq w_-^{-1}((c+C)/y)$,
for all $y>0$. This shows that 
$$
P((X_T^{\theta(n),z}-B)_-\geq c_n)\leq P(u_-((X_T^{\theta(n),z}-B)_-))\geq u_-(c_n))\leq  
w_-^{-1}((c+C)/u_-(c_n))\to 0,
$$
as $n\to\infty$, since $u_-(x)\to \infty$, $x\to\infty$ and $w_-^{-1}(q)\to 0$, $q\to 0$.

We claim that, for all $1\leq j\leq k$,
$P(|\theta_j(n)\Delta S_j|\geq c_n)\to 0$, $n\to\infty$ for any
sequence $c_n\to\infty$, $n\to\infty$. Indeed, fix $\varepsilon>0$.
$P(|\Delta S_j|\leq s)\geq 1-\varepsilon/2$ for $s$ large enough. Also, 
for $n$ large enough,
$P(|\theta_j(n)|\leq c_n /s)\geq 1-\varepsilon/2$.
Hence $P(|\theta_j(n)\Delta S_j|\geq c_n)\leq \varepsilon$ for $n$ large enough, as claimed.
It follows that also
$$
l(n):=P(|X_k^{\theta(n),z}|\geq c_n/3)\leq \sum_{j=1}^k P(|\theta_j(n)\Delta S_j|\geq c_n/(3k))\to 0,\quad n\to\infty.
$$

Lemma \ref{makk} implies that 
\begin{eqnarray*}
P(|B|\geq c_n/3)+l(n)+P((X_T^{\theta(n),z}-B)_-\geq c_n/3) &\geq&\\
P((\sum_{j=k+1}^T \theta_j(n)\Delta S_j)_-\geq c_n) &\geq&\\
P(\theta_{k+1}(n)\Delta S_{k+1}\leq -\kappa_{k} |\theta_{k+1}(n)|,\ \theta_j(n)\Delta S_j\leq 0,\
k+2\leq j\leq T, & &\\ 
|\theta_{k+1}(n)|\geq c_n/\kappa_{k}, D_{t}\neq \{0\}) &\geq&\\ 
E[1_{D_{t}\neq \{0\}}1_{|\theta_{k+1}(n)|\geq c_n/\kappa_{k}}\pi_{k+1}],
\end{eqnarray*}
so $E[1_{D_{t}\neq \{0\}}1_{|\theta_{k+1}(n)|\geq c_n/\kappa_{k}}\pi_{k+1}]\to 0$, $n\to\infty$.
Define $Q\sim P$ by $dQ/dP:=\pi_{k+1}/E\pi_{k+1}$. Since $\theta(n)\in\hat{\Phi}$,
$\theta_{k+1}(n)=0$ on $\{D_{t}=\{0\}\}$. It follows that 
$Q(|\theta_{k+1}(n)|\geq c_n/\kappa_{k})\to 0$, $n\to\infty$, which implies $Q(|\theta_{k+1}(n)|\geq c_n)\to 0$ 
and hence also 
$P(|\theta_{k+1}(n)|\geq c_n)\to 0$, $n\to\infty$. The induction step is completed.
\end{proof}

\noindent{\em Proof of Theorem \ref{masodik}.} Let $\theta(j)$ be such that 
$V(\theta(j),z)> \sup_{\theta\in\Phi}
V(\theta,z)-1/j$, $j\in\mathbb{N}$. By Remark \ref{projection} we may and will suppose that $\theta(j)\in\hat{\Phi}$ 
for all $j$. By \eqref{cccc}, the supremum is at least $V(0,z)>-\infty$ so 
$\inf_j V(\theta(j),z)>-\infty$, hence Lemma \ref{maine}
shows the tightness of the sequence of random variables
$$
(S_1,\ldots,S_T,\theta_1(j),\ldots,\theta_T(j)),\ j\in\mathbb{N}.
$$
Following verbatim the proof of Theorem 7.4 in \cite{cr} we get $\theta^*\in\Phi$ such that
(along a subsequence) 
$$
(S_1,\ldots,S_T,\theta_1(j),\ldots,\theta_T(j))\to (S_1,\ldots,S_T,\theta_1^*,\ldots,\theta_T^*)
$$
in law as $j\to\infty$, hence also $X_T^{\theta(j),z}$ converge to $X_T^{\theta^*,z}$ in law, $j\to\infty$.

As $u_{\pm},w_{\pm}$ are continuous, $w_{\pm}(P(u_{\pm}([X^{\theta(n),z}-B]_{\pm})\geq y))$
tend to $w_{\pm}(P(u_{\pm}([X^{\theta^*,z}-B]_{\pm})\geq y))$ outside the discontinuity points
of the cumulative distribution functions of $u_{\pm}([X^{\theta^*,z}-B]_{\pm})$, in particular,
for Lebesgue-a.e. $y$. Fatou's lemma implies 
$$
\limsup_{j\to\infty}V(\theta(j),z)\leq 
V(\theta^*,z),
$$
which shows that $\theta^*$ satisfies \eqref{problem}.
\hfill $\Box$

\section{A counterexample}

Let us define $\mathcal{P}:=\{\mathrm{Law}(X^{\phi,0}_T):\phi\in\Phi\}$.
The proof of Theorem \ref{masodik} consisted of two steps: first, the relative compactness
of the sequence of optimisers (for the weak convergence of probability measures) was shown using Lemma \ref{maine}; 
second,
the closedness of the set $\mathcal{P}$ was established referring to the proof of Theorem 7.4 in \cite{cr},
under Assumption \ref{AAA}.

In this section we provide an example which shows that the latter closedness property can
easily fail unless additional assumptions (such as Assumption \ref{AAA} of the present paper or Assumption 6.1
of \cite{cr}) are made. 
This fact is surprising since the set $\{X^{\phi,0}_T:\phi\in\Phi\}$
is closed in probability, even without (NA), see Proposition 2 of \cite{stricker} and Proposition 6.8.1
of \cite{ds}.

Our example will be a one-step model with one risky asset and a non-trivial 
initial sigma-algebra. Let
$U$ be uniform on $[0,1]$ and let $Y$ be a $\mathbb{Z}$-valued random variable, independent of $U$,
with $P(Y=-1)=1/2$, $P(Y=k)=1/2^{k+1}$, $k\geq 1$. Define $\mathcal{F}_0:=\sigma(U)$,
$\mathcal{F}_1:=\sigma(U,Y)$. Set $S_0=0$, $S_1=\Delta S_1:=-1$ if $Y=-1$
and $\Delta S_1=f_k(U)$ if $Y=k$, $k\geq 1$ where $f_k(x):=3^k+1/2+q_k(x)$, $x\in [0,1]$ and
$q_k$ is a complete orthogonal system in the Hilbert space
$$
\{h\in L^2([0,1],\mathcal{B}([0,1]),\mathrm{Leb}):\int_0^1 h(x)dx=0\}
$$ 
such that
each $q_k$ is continuous and $|q_k(x)|\leq 1/2$, $x\in [0,1]$. Such a system can easily be
constructed e.g. from the trigonometric system. This model clearly satisfies (NA)  but we claim that 
$$
\mathcal{P}=\{\mathrm{Law}(\phi \Delta S_1):\phi\mbox{ is }\mathcal{F}_0\mbox{-measurable}\}
$$
is \emph{not} closed for weak convergence.

We first construct a certain limit point for a sequence in $\mathcal{P}$. A ``creation of more
randomness'' takes place in the next lemma.
\begin{lemma}
 Define $g_n(x):=n(x-k/n)$, $k/n\leq x<(k+1)/n$, $k=0,\ldots,n-1$ and set $g_n(1)=1$, for $n\in\mathbb{N}$.
 We claim that $\mu_n:=\mathrm{Law}(U,g_n(U))$ converges weakly to 
 $\mu_{\infty}:=\mathrm{Law}(U,V)$, $n\to\infty$,
 where $V$ is uniform on $[0,1]$ and it is independent of $U$.
\end{lemma}
\begin{proof}
 It suffices to prove that, for all
 $0\leq a,b\leq 1$, we have $\mu_n([0,a]\times [0,b])\to \mu([0,a]\times [0,b])$,
 see Theorem 29.1 in \cite{bill}. Fix $a,b$ and define, for all $n$, $l(n)$  as the
largest integer with $l(n)/n\leq a$.
By the definition of $g_n$, we have that, for all $n\in\mathbb{N}\cup \{\infty\}$,
$$
\mu_n([0,l(n)/n]\times [0,b])=bl(n)/n.
$$
It is also clear that $\mu_n([0,a]\times [0,b])-\mu_n([0,l(n)/n]\times [0,b])\leq 1/n$
holds for all $n\in\mathbb{N}$ and also 
$\mu_{\infty}([0,a]\times [0,b])-\mu_{\infty}([0,l(n)/n]\times [0,b])\leq 1/n$,
hence $\mu_n([0,a]\times [0,b])\to \mu_{\infty}([0,a]\times [0,b])$, $n\to\infty$.
 \end{proof}
 
Define $\phi_n:=g_n(U)+1$, $n\in\mathbb{N}$. It follows that the sequence of triplets $(U,Y,\phi_n)$, $n\in\mathbb{N}$ converges
to $(U,Y,V)$ in law, where $V$ is uniform on $[1,2]$ and independent of $(U,Y)$.
Define $\bar{\phi}:=V$. Note that equipping $\mathbb{Z}$ with the discrete topology, $\Delta S_1$
is a continuous function of $(U,Y)$, hence, by the continuous mapping theorem, 
$\mathrm{Law}(\phi_n\Delta S_1)$
converges weakly to $\nu:=\mathrm{Law}(\bar{\phi}\Delta S_1)$. We claim, however, that
$\nu\notin\mathcal{P}$. 

Aguing by contradiction, let us suppose the existence of a Borel-function $g$ such that with $\phi:=g(U)$
one has $\nu=\delta:=\mathrm{Law}(\phi\Delta S_1)$. Let $s$ denote the support of the law of $\phi$.
If $s\cap (-\infty,0)\neq \emptyset$ then the support of $\delta$ would be unbounded from below
hence it cannot be equal to $\nu$. Hence $\phi\geq 0$ a.s. and then $-s=\mathrm{supp}(\nu)\cap (-\infty,0]=[-2,-1]$,
so $s=[1,2]$.

This implies that the following (a.s.) equalities hold between events:
$$
A_k:=\{\phi \Delta S_1\in [3^k, 2\times 3^k +2]\}=\{Y=k\}
=\{\bar{\phi} \Delta S_1\in [3^k, 2\times 3^k +2]\},
$$
for all $k\geq 1$. Then, by independence of $U$ from $Y$ and $V$ of $(U,Y)$,
\begin{eqnarray*}
E[\phi \Delta S_1 1_{A_k}]=P(Y=k)E[g(U) f_k(U)]=(1/2^{k+1})\int_0^1 g(x) f_k(x)dx &=&\\
E[\bar{\phi} \Delta S_1 1_{A_k}]=(1/2^{k+1})EV Ef_k(U)=(1/2^{k+1})(3/2)(3^k+1/2).
\end{eqnarray*}
It follows that, for all $k$,
$$
c_k:=\int_0^1 g(x)q_k(x)dx=E[g(U) (f_k(U)-3^k-1/2)]=[3/2-Eg(U)](3^k+1/2).
$$
Since the $q_k$ are orthogonal and uniformly bounded, necessarily $\sum_{k=1}^{\infty} c_k^2<\infty$, so $c_k=0$
for all $k$. This implies $Eg(U)=3/2$. By the completeness of the sequence $q_k$ we also
get that $g$ is a.s. constant, so $\phi=3/2$. This contradiction with $s=[1,2]$ shows our claim.


\begin{thebibliography}{30}

\bibitem{bkp}
A.~B. Berkelaar, R.~Kouwenberg, and T.~Post.
\newblock Optimal portfolio choice under loss aversion.
\newblock \emph{Rev. Econ. Stat.}, 86:\penalty0 973--987, 2004.

\bibitem{bg}
C.~Bernard and M.~Ghossoub.
\newblock Static portfolio choice under cumulative prospect theory.
\newblock \emph{Mathematics and Financial Economics}, 2:\penalty0 277--306,
  2010.

\bibitem{bill}
{P. Billingsley.}
\newblock \emph{Probability and measure.}
\newblock John Wiley \& Sons, New York, 1979.


\bibitem{cp}
L.~Carassus and H.~Pham.
\newblock Portfolio optimization for nonconvex criteria functions.
\newblock \emph{RIMS K\^{o}kyuroku series, ed. Shigeyoshi Ogawa},
  1620:\penalty0 81--111, 2009.

\bibitem{cr}
L. Carassus and M. R\'asonyi.
\newblock On Optimal Investment for
a Behavioural Investor in Multiperiod Incomplete Market Models.
\newblock \emph{Forthcoming in Math. Finance, published online}, 2014.
\newblock \texttt{doi:10.1111/mafi.12018}


\bibitem{nc}
L.~Carassus and M.~R\'asonyi.
\newblock Maximization for non-concave utility functions in discrete-time
  financial market models.
\newblock \emph{Submitted}, 2014.
\newblock \texttt{arXiv:1302.0134}

\bibitem{cd}
G.~Carlier and R.-A. Dana.
\newblock Optimal demand for contingent claims when agents have law invariant
  utilities.
\newblock \emph{Math. Finance}, 21:\penalty0 169--201, 2011.

\bibitem{ds}
F.~Delbaen and W.~Schachermayer.
\newblock {\em The mathematics of arbitrage}.
\newblock Springer-Verlag, Berlin, 2006.


\bibitem{hz}
X.~He and X.~Y. Zhou.
\newblock Portfolio choice under cumulative prospect theory: An analytical
  treatment.
\newblock \emph{Management Science}, 57:\penalty0 315--331, 2011.

\bibitem{jz}
H.~Jin and X.~Y. Zhou.
\newblock Behavioural portfolio selection in continuous time.
\newblock \emph{Math. Finance}, 18:\penalty0 385--426, 2008.

\bibitem{ks}
Yu. M. Kabanov and Ch. Stricker.
\newblock A teacher's note on no-arbitrage criteria.
\newblock \emph{In: S\'eminaire de Probabilit\'es, vol. XXXV}, 149--152, Springer-Verlag, 2001.

\bibitem{kt}
D.~Kahneman and A.~Tversky.
\newblock Prospect theory: An analysis of decision under risk.
\newblock \emph{Econometrica}, 47:\penalty0 263--291, 1979.

\bibitem{kallenberg}
O. Kallenberg.
\newblock \emph{Foundations of modern probability.}
\newblock 2nd edition, Springer-Verlag, Berlin, 2002.

\bibitem{menger34}
K.~Menger. \newblock {D}as {U}nsicherheitsmoment in der {W}ertlehre.
\newblock \emph{Zeitschrift
  f\"{u}r National\"{u}konomie}, 5:459--485, 1934.

\bibitem{muravievrogers13}
R.~Muraviev and {L.\,C.\,G.} Rogers. 
\newblock{Utilities bounded below.} 
\newblock \emph{Ann.\ Finance}, {9}:271--289, 2013.
  
  \bibitem{pennanen1}
  T. Pennanen. 
  \newblock Optimal investment and contingent claim valuation in illiquid markets.
  \newblock\emph{Forthcoming in Finance Stoch.}, 2014.
  
  
  
  \bibitem{rr} M. R\'asonyi, A. M. Rodrigues. 
  \newblock Optimal portfolio choice for a behavioural investor in continuous
time. 
\newblock\emph{Ann. Finance}, 9:291--318, 2013. 

\bibitem{rr2} M. R\'asonyi, A. M. Rodrigues. Continuous-Time Portfolio Optimisation 
for a Behavioural Investor with Bounded Utility on Gains. \emph{Electronic Communications in Probability} 
vol. 19, article no. 38, 1--13, 2014. 

\bibitem{rrv} M. R\'asonyi, J. G. Rodr\'\i guez-Villarreal.
\newblock Optimal investment under behavioural criteria -- a dual approach. 
\newblock\emph{Forthcoming in Banach Center Publications}, 2014.
\newblock\texttt{arXiv:1405.3812} 


\bibitem{rs}
M.~R\'asonyi and L.~Stettner.
\newblock On the utility maximization problem in discrete-time financial market
  models.
\newblock \emph{Ann. Appl. Probab.}, 15:\penalty0 1367--1395, 2005.








 
\bibitem{reichlinthesis}
C.~Reichlin, \emph{Non-concave utility maximization: optimal investment,
  stability and applications.} 
  \newblock Ph.D.\ thesis, ETH Z\"{u}rich, 2012. 
  \newblock {Diss. ETH No. 20749}.

  \bibitem{r12}
C.~Reichlin.
\newblock Utility maximization with a given pricing measure when the utility is
  not necessarily concave.
\newblock \emph{Mathematics and Financial Economics}, 7:\penalty0 531--556,
  2013.

  
\bibitem{ryan}  T. M. Ryan. \newblock The use of unbounded utility functions in expected-utility maximization:
comment. 
\newblock\emph{Quart. J. Econom.} 88:133--135, 1974.

\bibitem{stricker}
{Ch. Stricker.}
\newblock Arbitrage et lois de martingale.
\newblock \emph{Annales de l’Institut Henri Poincar\'e, section B}, 26:451--460, 1990.


\bibitem{tk}
A.~Tversky and D.~Kahneman.
\newblock Advances in prospect theory: Cumulative representation of
  uncertainty.
\newblock \emph{J. Risk \& Uncertainty}, 5:\penalty0 297--323, 1992.


\end{thebibliography}
\end{document}